\documentclass[%
reprint,
 amsmath,amssymb,
 aps,
 prl,showpacs
]{revtex4-1}

\usepackage{graphicx}
\usepackage{dcolumn}
\usepackage{bm,bbm}
\usepackage{braket}
\usepackage{amsthm,amsmath,amssymb,amsbsy}
\usepackage{hyperref}
\usepackage{times,txfonts}

\newtheorem{theorem}{Theorem}
\newcommand{\norm}[1]{\lVert#1\rVert}

\DeclareMathOperator{\SL}{SL}
\begin{document}

\title{$\mbox{\!\!Entanglement quantification made easy: Polynomial measures invariant under convex decomposition}$}

\author{Bartosz Regula}
\affiliation{$\mbox{School of Mathematical Sciences, The University of Nottingham, University Park,
Nottingham NG7 2RD, United Kingdom}$}
\author{Gerardo Adesso}
\affiliation{$\mbox{School of Mathematical Sciences, The University of Nottingham, University Park,
Nottingham NG7 2RD, United Kingdom}$}

\date{\today}

\begin{abstract}
Quantifying entanglement in composite systems is a fundamental challenge, yet exact results are only available in few special cases. This is because hard optimization problems are routinely involved, such as finding the convex decomposition of a mixed state with the minimal average pure-state entanglement, the so-called convex roof. We show that under certain conditions such a problem becomes trivial. Precisely, we prove by a geometric argument that polynomial entanglement measures of degree $2$ are independent of the choice of pure-state decomposition of a mixed state, when the latter has only one pure unentangled state in its range. This allows for the analytical evaluation of convex roof extended entanglement measures in classes of rank-two states obeying such condition. We give explicit examples for the square root of the three-tangle in three-qubit states, and show that several representative classes of four-qubit pure states have marginals that enjoy this property.
\end{abstract}

\maketitle

Entanglement is an emblem of quantum mechanics and the most important component for a broad spectrum of quantum technologies \cite{horodecki_2009,amico_2008}. The more a quantum state is entangled, the better will it perform in an information processing and communication task, compared to any unentangled state \cite{wootters_1998,piani_2009,boixo_2008}. Quantifying entanglement exactly is therefore a significative requirement to develop a rigorous assessment of nonclassical enhancements in realistic applications \cite{eltschka_2014}. With the advent of quantum information theory in the last two decades, a sound machinery has been developed for the characterization and quantification of entanglement as a resource \cite{vedral_1997,wootters_1998,horodecki_2009,virmani_2007,eltschka_2014}.

An entanglement measure $E(\ket{\psi})$ defined on pure quantum states $\ket{\psi}$ is a positive real function which is 0 iff $\ket{\psi}$ is separable. Additionally, every such measure has to be an \textit{entanglement monotone}, that is, it cannot increase on average under local operations and classical communication (LOCC) \cite{vedral_1997}. One of the difficulties of quantifying entanglement lies in the fact that entanglement measures as defined above do not typically admit an easy way to extend their scope to all mixed quantum states. The so-called \textit{convex roof} of a measure of entanglement $E(\ket{\psi})$ is obtained by finding the largest convex function on the set of mixed states which corresponds to $E$ on pure states \cite{uhlmann_1998, uhlmann_2010}. One can use this construction to define the extension of an entanglement measure $E$ to mixed states $\rho$ as
\begin{equation}\label{eq:convexroof}E(\rho) = \min_{\{p_i, \ket{\psi_i}\}} {\sum}_i p_i E(\ket{\psi_i})\,,\end{equation}
where the minimization is performed over all sets $\{p_i, \ket{\psi_i}\}$ such that $\sum_i p_i \ket{\psi_i}\bra{\psi_i} = \rho$, that is, all convex decompositions of $\rho$ into pure states, with normalized weights  $\sum_i p_i = 1$. With this definition, $E$ is guaranteed to remain an entanglement monotone over mixed states as well \cite{vidal_2000}. It is however a formidable problem to find the optimal decomposition in Eq.~(\ref{eq:convexroof}), hence to obtain an analytical form for $E(\rho)$ in general \cite{gurvits_2003,huang_2014,toth_2015}. Notable cases in which this has been accomplished include the evaluation of the entanglement of formation in all two-qubit states (in terms of concurrence) \cite{hill_1997,wootters_1997}, in highly symmetric two-qudit states \cite{terhal_2000}, and in symmetric two-mode Gaussian states \cite{giedke_2003}, as well as the computation of the three-tangle \cite{coffman_2000} in particular families of three-qubit mixed states \cite{lohmayer_2006,eltschka_2008,jung_2009,siewert_2012,viehmann_2012,eltschka_2012_srep}.

A particularly useful class of functions to consider for entanglement quantification are the \textit{polynomial invariants}, that is, polynomial functions in the coefficients of a pure state $\ket{\psi}$ which are invariant under stochastic  LOCC (SLOCC). For a system of $m$ qudits, a polynomial invariant of homogeneous degree $h$ is therefore a function $P$ which satisfies
\begin{equation}P(\kappa\, L \ket{\psi}) = \kappa^h\, P(\ket{\psi})\,,\label{eq:sleq}\end{equation}
for a constant $\kappa >0$ and an invertible linear operator $L\in \SL(d,\mathbb{C})^{\otimes m}$ representing the SLOCC transformation \cite{dur_2000}. The absolute value of any such polynomial with $h\leq 4$ defines in fact an entanglement monotone \cite{verstraete_2003, eltschka_2012}. Two common monotones, the concurrence for two qubits \cite{hill_1997} and the three-tangle for three \cite{coffman_2000}, are obtained in this way. Out of all possible homogeneous degrees $h$,  degree $2$ is of particular significance, as only then the SLOCC invariance of a polynomial entanglement measure $E(\ket{\psi})$ extends to its convex roof $E(\rho)$ \cite{viehmann_2012}.

One can relate the properties of entanglement measures with the geometric representation of quantum states on a hypersphere  (a generalization of the Bloch sphere) by considering the so-called {\em zero polytope} \cite{lohmayer_2006,osterloh_2008}. For a given state $\rho$ and an entanglement measure $E$, it is defined as the convex hull of all pure states (spanned by the eigenvectors of $\rho$) with vanishing $E$. Since entanglement vanishes for any convex mixture of such states but will never vanish for a state lying outside of the convex hull, the zero polytope gives a useful visual representation of the region of the hypersphere with zero entanglement. Various methods for constructing bounds to polynomial entanglement measures in mixed states rely on finding  states within the zero polytope and using them to form suitable convex combinations with states outside of it \cite{lewenstein_1998, acin_2001, rodriques_2014}.

In this Letter, we analyze the particular situation when the zero polytope for a given state $\rho$ is reduced to a single point, that is, when there is only one state in the range of $\rho$ with vanishing entanglement. Our investigation is naturally specialized to rank-$2$ states $\rho$, which admit an intuitive geometrical representation on a Bloch sphere, and represent ideal testbeds to analyze structural properties of multipartite entanglement in mixed states. We show that, when the zero polytope reduces to a point, any polynomial entanglement measure of degree 2 simply corresponds to a measure of distance on the Bloch sphere.  This property in turn renders the value of the measure {\em independent of the decomposition} of $\rho$ into pure states. Therefore, the convex roof extension of the entanglement measure for  $\rho$ becomes trivial, and
can be evaluated analytically in any decomposition, e.g.~in the spectral one. 

Although geometric methods have long been known to be valuable tools in quantum information theory \cite{bengtsson_2007}, this surprising result relies only on classical Euclidean geometry, which does find use in the study of quantum correlations \cite{jevtic_2014,milne_2014}, but whose specific application to entanglement quantification went seemingly unnoticed so far.

We begin with the formal definition of the zero polytope. Given a state $\rho$ of rank $\eta$ and an entanglement measure $E(\ket{\psi})$ which is the absolute value of a polynomial of degree $h$ in the coefficients of the pure state $\ket{\psi}$, we can write the equation
\begin{equation}\label{eq:polytope}
E\left( {\sum}_{j=0}^{\eta-1} \omega_j \ket{\phi_j} \right) =0\,,\end{equation}
where $\omega_j$ are complex coefficients and $\ket{\phi_j}$ are the eigenvectors of $\rho$. The zero polytope is then defined to be the convex hull of all pure states which satisfy the above equation \cite{osterloh_2008}.

Noting that the expression in Eq.~\eqref{eq:polytope} is in fact a polynomial of degree $h$ in the coefficients $\omega_j$, we are interested in the case when the polynomial has a unique root in $\omega_j$, that is, there is only one such state $\ket{z}$ defined as a linear combination of the eigenvectors $\{\ket{\phi_j}\}$ such that $E(\ket{z})$ vanishes. States $\rho$ for which this happens will be labeled as {\em one-root} states (shorthand for ``one root to rule them all''). Since multivariate complex polynomials have uncountably infinite sets of solutions and there is no straightforward method to investigate their roots \cite{osterloh_2008}, our analysis is limited to $\eta=2$. In this case we can represent the rank-$2$ state $\rho$ as a point (or Bloch vector) $\mathbf r \in \mathbb{R}^3$ in the standard Bloch sphere, with polar points corresponding to the eigenvectors $\ket{\phi_0}$ and $\ket{\phi_1}$ of $\rho$, see Fig.~\ref{fig:sphere}. We then have, up to normalization, that the root state can be written as $\ket{z} = \ket{\phi_0} + z \ket{\phi_1}$ (assuming $E(\ket{\phi_1}) \neq 0$) for some $z \in \mathbb{C}$, while for any pure state $\ket{\omega} = \ket{\phi_0} + \omega \ket{\phi_1}$ one has
\begin{equation}\label{eq:distance}
E\big( \ket{\phi_0} + \omega \ket{\phi_1} \big) = N \left|\omega - z\right|^h\,,\end{equation}
where $N$ is a normalization factor. We note that for $h=2$ this expression is proportional to the squared Euclidean distance   $\norm{\boldsymbol{\omega} - \mathbf z}^2$ between the Bloch vectors associated to $\ket{\omega}$ and $\ket{z}$. This interpretation of the entanglement measure as a \textit{de facto} measure of distance allows us to employ the geometrical properties of Euclidean spaces to investigate the behaviour of $E$ --- which from now on will precisely denote a polynomial entanglement measure of degree 2 --- on an arbitrary {\it mixed} state $\rho$ with Bloch vector $\mathbf{r}$ inside the sphere. We are thus ready to present our central result, which establishes a geometric relation for $E$ across all possible decompositions of $\rho$.

\begin{theorem}\label{geothm}Consider an $n$-sphere with radius $R$ and center located at $\mathbf{o}$. We will indicate by $\{\lambda_i, \mathbf p_i\}$ a finite set of points $\mathbf p_i \in \mathbb{R}^{n+1}$ on the sphere with corresponding weights $\lambda_i$, normalized so that $\sum_i \lambda_i = 1$. Let us choose a particular point $\mathbf z$ on the sphere and denote by $\{\alpha_i, \mathbf{a}_i\}$ a set such that all points $\{\mathbf{a}_i\}$ are equidistant from $\mathbf z$, that is $\|{\mathbf z- \mathbf{a}_i}\| = \|{\mathbf z- \mathbf{a}_k}\| \;\forall\, i,k$, according to the standard Euclidean distance. Let $\mathbf{g} $ denote the  barycenter of the family of points, that is,
$\mathbf{g}  = \sum_i \alpha_i \mathbf{a}_i$.
Then, for any other set of points $\{\beta_j, \mathbf{b}_j\}$ which lie on the same sphere and share the same barycenter $\mathbf{g}$, the following holds:
\begin{equation}\label{invariance}
\mbox{$\sum_j \beta_j \|{\mathbf z - \mathbf{b}_j}\|^2 = \|{\mathbf z- \mathbf{a}_l}\|^2$, \quad $\forall \  \mathbf{a}_l \in \{\mathbf{a}_i\}$}.
\end{equation}
\end{theorem}
\begin{proof}
Apollonius' formula \cite{berger_1987} says that for any set of points $\{\alpha_i, \mathbf{a}_i\}$ with barycenter $\mathbf{g} $ and any chosen point $\mathbf z \in \mathbb{R}^{n+1}$ we can write: $
\sum_i \alpha_i \norm{\mathbf z- \mathbf{a}_i}^2 = \norm{\mathbf z-\mathbf{g} }^2 + \sum_i \alpha_i \norm{\mathbf{g}  - \mathbf{a}_i}^2$,
but since we chose $\{\alpha_i, \mathbf{a}_i\}$ to be equidistant from $\mathbf z$, we get
\begin{equation}\label{rho}
\mbox{$\norm{\mathbf z-\mathbf{g} }^2 =  \norm{\mathbf z-\mathbf{a}_l}^2 - \sum_i \alpha_i \norm{\mathbf{g} -\mathbf{a}_i}^2$},
\end{equation}
for any chosen $\mathbf{a}_l$. Applying the same formula to any other set of points $\{\beta_j, \mathbf{b}_j\}$  with the same barycenter $\mathbf{g} $ gives:
\begin{align}
&\mbox{$\sum_j \beta_j \norm{\mathbf z-\mathbf{b}_j}^2 = \norm{\mathbf z-\mathbf{g} }^2 + \sum_j \beta_j \norm{\mathbf{g} -\mathbf{b}_j}^2$} \nonumber \\
\label{almost} &=\mbox{$\norm{\mathbf z-\mathbf{a}_l}^2 - \sum_i \alpha_i \norm{\mathbf{g} -\mathbf{a}_i}^2 + \sum_j \beta_j \norm{\mathbf{g} -\mathbf{b}_j}^2$},\end{align}
by Eq.~(\ref{rho}). Since the points $\{\alpha_i, \mathbf{a}_i\}$ lie on the $n$-sphere, we get
$R^2 = \norm{\mathbf{a}_i - \mathbf{o}}^2 = \norm{\mathbf{a}_i}^2 - 2 \mathbf{a}_i \cdot \mathbf{o} + \norm{\mathbf o}^2$
for each $\mathbf{a}_i$, where $\cdot$ denotes the standard Euclidean inner product. We can average over this expression with the weights $\{\alpha_i\}$ to obtain
$R^2 = \big(\sum_i \alpha_i \norm{\mathbf{a}_i}^2\big) - 2 \left( \sum_i \alpha_i \mathbf{a}_i \right) \cdot \mathbf o + \norm{\mathbf o}^2$, that is,
$\sum_i \alpha_i \norm{\mathbf{a}_i}^2 = R^2 + 2 \mathbf{g}  \cdot \mathbf o - \norm{\mathbf o}^2$.
Since the points $\{\beta_j, \mathbf{b}_j\}$ lie on the same $n$-sphere and share the same barycenter $\mathbf{g}$, we can conclude that
$\sum_i \alpha_i \norm{\mathbf{a}_i}^2 = \sum_j \beta_j \norm{\mathbf{b}_j}^2$,
and therefore
$\sum_i \alpha_i \norm{\mathbf{g} -\mathbf{a}_i}^2 = \sum_j \beta_j \norm{\mathbf{g} -\mathbf{b}_j}^2$.
Using this in Eq.~(\ref{almost}) gives the final result announced in Eq.~(\ref{invariance}).
\end{proof}

Theorem~\ref{geothm} implies that, for every one-root state $\rho$ with Bloch vector $\mathbf{r}$ (identifying $\mathbf{r}$ with the barycentre $\mathbf{g}$ in the statement of the Theorem), any polynomial entanglement measure $E$ of degree 2 has the same value irrespective of the chosen decomposition of $\rho$ into a set of $\nu \geq 2$ pure states. That is, the measure $E$ is an {\em affine} function on the whole Bloch sphere,
\begin{equation}\label{affine}\mbox{$E\left({\sum}_i p_i \ket{\psi_i}\bra{\psi_i}\right) = {\sum}_i p_i E\left(\ket{\psi_i}\right)$}\,, \quad \forall\ \{p_i, \ket{\psi_i}\}_{i=0}^{\nu-1}\,.\end{equation}
The evaluation of $E(\rho)$ is thus {\it made easy}, and can be carried out exactly in any decomposition. It is particularly instructive to consider a decomposition of $\rho$ such that all pure states $\{\ket{\psi_i}\}$ with Bloch vectors $\{\mathbf{a}_i\}$ (in the notation of Theorem~\ref{geothm}) lie on the  secant plane equidistant from the root point $\mathbf z$, see Fig.~\ref{fig:sphere}. The value of $E(\rho)$ then corresponds to the squared distance from any of these points $\mathbf{a}_i$ to $\mathbf z$, according to Eq.~(\ref{eq:distance}). But since this can be done for any other state with a Bloch vector lying in the same plane, the equidistant plane is in fact a plane of constant entanglement.

Let us now introduce the radial state $\rho_c$ whose  vector $\mathbf{c}$ is at the centre of the secant plane, i.e.~at the intersection between the plane and the Bloch sphere diameter joining $\mathbf z$ with the antipodal point $\mathbf{z'}$. The latter point corresponds to the pure state $\ket{z'}$ with maximal entanglement $E$ on the sphere. As $\rho$ and $\rho_c$ are on the same secant plane, one has $E(\rho) = E(\rho_c)$. The latter can be evaluated by exploiting the affinity of $E$, Eq.~(\ref{affine}), and taking the decomposition  $\rho_c = \frac{1}{2}\norm{\mathbf{c}- \mathbf{z'}} \ket{z}\bra{z} + \frac{1}{2}\norm{\mathbf{c} - \mathbf{z}} \ket{z'}\bra{z'}$.  Since $E(\ket{z})=0$, we finally get, for any one-root state $\rho$ in the Bloch sphere, the closed formula
\begin{align}\label{eq:entdistance}
E(\rho) = E(\rho_c) &= \mbox{$\frac{1}{2}$}\norm{\mathbf{c} -\mathbf z}\, E(\ket{z'}) =D_{\rm Tr}( \rho_c,   \ket{z}\!\bra{z}) \, E(\ket{z'}),
\end{align}
where one recognizes the trace distance $D_{\rm Tr}(\rho, \tau) = \frac12 {\rm Tr} | \rho-\tau|$ in the second expression.
The connection with Eq.~(\ref{eq:distance}) is made explicit by using elementary Euclidean geometry, which yields $E(\rho) = E(\ket{\psi_l}) = N \| \mathbf{a}_l - \mathbf{z} \|^2 = 2N \|\mathbf{c}-\mathbf{z}\| = E(\rho_c)$\,,
for any index $l \in \{0,\ldots,\nu-1\}$. Comparing with Eq.~(\ref{eq:entdistance}) we find the value of  the normalization constant, $N = E(\ket{z'})/4$.

\begin{figure}[t]
\centering
\includegraphics[width=7.5cm]{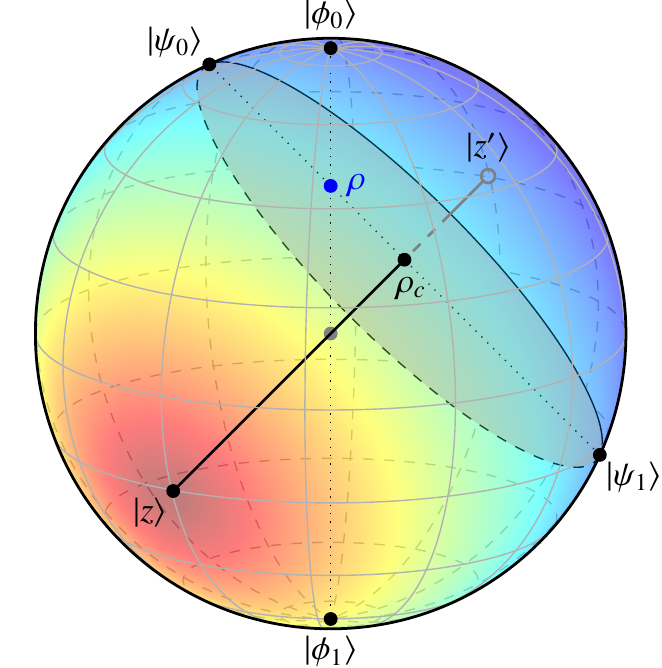}
\caption{(Color online) The entanglement of a rank-2 one-root state $\rho$, with detail of an equidistant decomposition  $\rho = p \ket{\psi_0}\bra{\psi_0} + (1-p) \ket{\psi_1}\bra{\psi_1}$, visualized on the Bloch sphere with the polar states $\{\ket{\phi_0}, \ket{\phi_1}\}$ being the eigenvectors of $\rho$. The flat shaded surface is the plane of constant entanglement for $\rho$ and the radial state $\rho_c$. The coloring on the sphere corresponds to the value of entanglement, which amounts to the distance from the root state $\ket{z}$, with red meaning vanishing entanglement, and blue meaning maximum entanglement, reached on the antipodal state $\ket{z'}$. See text for further details.}\label{fig:sphere}
\end{figure}

To summarize, given two orthonormal states $\{\ket{\phi_0},\,\ket{\phi_1}\}\in \mathbb{C}^{2^m}$ of a $m$-qubit system, one can construct a family of rank-$2$ mixed states $\rho$ defined as
\begin{equation}\label{eq:blochvec}
\rho = \mbox{$\sum_{i,j=0}^{1} B_{ij} \ket{\phi_i}\bra{\phi_j}$}\,,
\end{equation}
with $B = \frac{1}{2} \left(\mathbbm{1} + \mathbf{r} \cdot \boldsymbol{\sigma} \right)$,
where $\boldsymbol{\sigma} = (\sigma_x, \sigma_y, \sigma_z)$ denotes the Pauli matrices and
$\mathbf{r} = r \  (\sin\theta \cos\varphi, \ \sin\theta \sin\varphi, \ \cos\theta)^T$ the Bloch vector of $\rho$, with $0 \leq r \leq 1$, $0 \leq \theta \leq \pi$, and $0 \leq \varphi \leq 2\pi$.
Given now a polynomial entanglement measure $E$ of degree $2$, then if $E(\ket{\phi_0} + z \ket{\phi_1})=0$ has only one root $z$, we say that any state $\rho$ of Eq.~(\ref{eq:blochvec}) has the one-root property, which implies (by Theorem~\ref{geothm}) that its convex roof extended entanglement measure $E(\rho)$ is independent of the pure-state decomposition of $\rho$ and exactly computable via Eq.~(\ref{eq:entdistance}).

In order to present concrete examples, we begin by investigating the concurrence ${\cal C}$ of two qubits \cite{hill_1997}, for which general analytical expressions are known \cite{hill_1997,wootters_1997} and can be  compared with the results presented here. Visualizing Fig.~\ref{fig:sphere}, we can always apply a change of basis to rotate the Bloch sphere such that the north and south poles are occupied by the root $\mathbf{z}$ and the antipodal point $\mathbf{z'}$, respectively. We can further identify the root state $\ket{z}$ with the computational product state $\ket{00}$. By imposing  $\langle z|z'\rangle=0$ and that  the concurrence vanish only on $\ket{z}$, we obtain a complete characterization of two-qubit one-root states (up to local unitaries), specified by the basis:
$\ket{\phi_0} \equiv \ket{z} = \ket{00}$ and
$\ket{\phi_1} \equiv \ket{z'} = \cos\left(\frac{\gamma}{2}\right) \ket{01} +  \sin\left(\frac{\gamma}{2}\right) e^{i \delta} \ket{10}$, with $0 \leq \gamma \leq \pi$, $0 \leq \delta \leq 2\pi$. All horizontal planes crossing the ball are surfaces of constant concurrence. For any two-qubit state $\rho$ inside the sphere, defined as in Eq.~(\ref{eq:blochvec}), the concurrence
computed from Eq.~\eqref{eq:entdistance} is ${\cal C}\left(\rho\right) = \mbox{$\frac{1}{2}$}\big| 1 - \braket{\phi_0 | \rho | \phi_0} + \braket{\phi_1 | \rho | \phi_1}\big|\ {\cal C}\left(\ket{\phi_1}\right)= \mbox{$\frac{1}{2}$}\left(1- r \cos\theta\right) \sin\gamma$,
which coincides with the known general solution from \cite{hill_1997}.

We now focus on the case of three qubits ($m=3$) and adopt the square root of the three-tangle $\sqrt{\cal T}$ \cite{coffman_2000,viehmann_2012} as a polynomial measure of tripartite entanglement, whose explicit formula for pure states $\ket{\psi} \in \mathbb{C}^8$ is provided in \cite{coffman_2000}. Such a measure plays a prominent role in studies of monogamy of entanglement \cite{coffman_2000, regula_2014}, yet at present no closed solution exists in general for its evaluation on mixed states, beyond a few special cases \cite{lohmayer_2006,eltschka_2008,jung_2009,siewert_2012,viehmann_2012,eltschka_2012_srep}.
We can readily construct a representative family of one-root rank-2 states of three qubits in a similar way as for two qubits.
We take the poles of the Bloch sphere to be, respectively, the generalized $W$ state
$\ket{\phi_0} \equiv \ket{z} =  a \ket{001} + b \ket{010} + c \ket{100}$
with $a^2+b^2+c^2=1$ (where $a,b,c$ are chosen real for ease of illustration), and the entangled state $\ket{\phi_1} \equiv \ket{z'} =  g \ket{000} + t_1 \ket{011} + t_2 \ket{101} + t_3 \ket{110} + e^{i\gamma}h \ket{111}$ with $g^2+h^2+\sum_i t_i^2 = 1$, $g \geq t_i$, $h \geq 0$, and $-\frac{\pi}{2}\leq \gamma \leq \frac{\pi}{2}$, as defined in \cite{tamaryan_2009}. Imposing the one-root property leads to  $h=0$ and $t_3 =(\sqrt{c t_1} + \sqrt{b t_2})^2/a$. We can then write any state $\rho$ inside the Bloch sphere as in Eq.~\eqref{eq:blochvec}, which leads us to the exact expression for the square root of three-tangle $\sqrt{\cal T}$ of $\rho$,
\begin{align}\sqrt{\cal T}\left(\rho\right) &= \mbox{$\frac{1}{2}$}\big| 1 - \braket{\phi_0 | \rho | \phi_0} + \braket{\phi_1 | \rho | \phi_1}\big|\ \sqrt{\cal T}\left(\ket{\phi_1}\right)\\
&= \sqrt{\left| \frac{g t_1 t_2}{a^9} \right|} \left| \sqrt{c t_1} + \sqrt{b t_2}\right| \left| 1 - r\cos\theta \right| \nonumber \\
& \quad \times \bigg|a^4+\left[\left(\sqrt{c t_1} + \sqrt{b t_2}\right)^4 + a^2 \left(g^2 + t_1^2 + t_2^2\right)\right]^2\bigg|.\nonumber \end{align}
Tripartite entanglement in this $7$-parameter class of three-qubit states $\rho$ has thus been effortlessly quantified, thanks to their one-root property and its geometric implications.

Beyond specific examples, one can wonder whether a more systematic characterization of one-root three-qubit states is possible, so as to gauge the relevant range of applicability of our exact results. The answer is affirmative. Notice that every rank-$2$ three-qubit state $\rho$ can be purified to a four-qubit state $\ket{\Psi} \in \mathbb{C}^{16}$, and conversely the set of marginals obtained by tracing out one qubit from arbitrary four-qubit pure states $\ket{\Psi}$ completely characterizes the set of rank-$2$ three-qubit states $\rho$. We can then aim to identify the one-root three-qubit states in terms of their four-qubit purifications. To this end, we recall that while pure four-qubit states have an infinite number of SLOCC-inequivalent classes \cite{dur_2000}, that is, subsets of states which cannot be transformed into one another by performing SLOCC operations, from the point of view of classifying their entanglement properties they can in fact be conveniently grouped into $9$ classes \cite{verstraete_2002}. Each of these classes then forms a subset $\Upsilon_\mu \subset \mathbb{C}^{16}$ (for $\mu=1,\ldots,9$) represented by a generating family $\ket{G_\mu}$ (dependent on at most four continuous complex parameters), such that all the states $\ket{\Psi_\mu} \in \Upsilon_\mu$ belonging to the $\mu^{\text{th}}$  class are constructed as $\ket{\Psi_\mu} = L \ket{G_\mu} / \|L \ket{G_\mu}\|$, where $L \in \SL(2, \mathbb{C})^{\otimes 4}$ is a SLOCC operation. The union of all the nine classes $\bigcup_{\mu=1}^9 \Upsilon_\mu$ covers the Hilbert space of {\em all} four-qubit pure states, up to permutations of the qubits.

We now make a useful observation.
When two pure states are SLOCC-equivalent, the ranges of their corresponding reduced subsystems are spanned by SLOCC-equivalent bases, which means that all states in the reduced ranges are related by an invertible linear transformation \cite{chen_2006,*chen_2006-1}. Since for any two SLOCC-equivalent states  the polynomial entanglement measure $E$ either vanishes on both or is strictly nonzero on both \cite{viehmann_2012}, we have that SLOCC operations preserve the number of zero-$E$ states in the ranges of the reduced subsystems. In other words, the number of roots in the zero polytope for the marginals of four-qubit states is a SLOCC-invariant.

It  then suffices to check the marginals of the generators $\ket{G_\mu}$ to look for the one-root property. This can be done analytically after some straightforward algebra, and we find as a result that {\em four} of the nine classes of four-qubit pure states have three-qubit marginals which can enjoy the one-root property. This, applied to the whole respective SLOCC classes \cite{verstraete_2002}, characterizes completely the set of  three-qubit one-root states, and entails that for all these rank-$2$ mixed states $\rho$ we can exactly calculate the convex roof extended entanglement measure $\sqrt{\cal T}$ thanks to Theorem~\ref{geothm}, which is remarkable.

Explicitly, the classes whose marginals are generally one-root are: class $4$ (tracing out any one of the four qubits), class $5$ (tracing out qubit $2$ or qubit $4$ only), and classes $7$ and $8$ (tracing out qubit $2$, $3$, or $4$ only).
The corresponding sets of three-qubit one-root states are given therefore by
\begin{align}\label{eq:rhoslocc}
&\rho^L_{\mu,k} = {\rm Tr}_k \left[\frac{L \ket{G_\mu}\!\bra{G_\mu}L^\dagger}{{\rm Tr}(L \ket{G_\mu}\!\bra{G_\mu}L^\dagger)}\right]\,, \quad \forall\ L \in \SL(2, \mathbb{C})^{\otimes 4}\,, \nonumber \\
& \qquad \quad \mbox{and }\ \forall\ k \in \left\{
                  \begin{array}{ll}
                    \{1,2,3,4\}, & \hbox{if $\mu = 4$;} \\
                    \{2,4\}, & \hbox{if $\mu=5$;} \\
                    \{2,3,4\}, & \hbox{if $\mu = 7$ or $8$.} \\
                  \end{array}
                \right.
\end{align}
For completeness, we report the relevant (unnormalized) generators:
$\ket{G_4}=a(\ket{0000}+\ket{1111}) + \frac{a+b}{2}(\ket{0101}+\ket{1010}) + \frac{a-b}{2}(\ket{0110}+\ket{1001}) + \frac{i}{\sqrt{2}}(\ket{0001} + \ket{0010} + \ket{0111} + \ket{1011}) \equiv \ket{L_{ab_3}}$,
$\ket{G_5}=  a(\ket{0000} + \ket{0101} + \ket{1010} + \ket{1111}) +i\ket{0001} + \ket{0110} - i\ket{1011} \equiv \ket{L_{a_4}}$,
$\ket{G_7}= \ket{0000} + \ket{0101} + \ket{1000} + \ket{1110} \equiv \ket{L_{0_{5\oplus\bar{3}}}}$, $\ket{G_8}= \ket{0000} + \ket{1011} + \ket{1101} + \ket{1110} \equiv \ket{L_{0_{7\oplus\bar{1}}}}$, where $a,b \in \mathbb{C}$ with $\text{Re}(a), \text{Re}(b) \geq 0$, and $\ket{L_{\bullet}}$ refers to the notation of \cite{verstraete_2002}.
The square root of three-tangle for all the states in Eq.~(\ref{eq:rhoslocc}) is given exactly by Eq.~(\ref{eq:entdistance}); if one prefers, it can also be evaluated numerically in any convex decomposition (e.g.~the spectral one), with no optimization required.

We finally note that many entanglement bounds for convex roof extended measures will be tight on the one-root states, because of their special properties. For instance, bounds such as the best separable approximation for two qubits \cite{lewenstein_1998}, the best $W$ approximation for three qubits \cite{acin_2001}, and the generalized best zero-$E$ approximation \cite{rodriques_2014} are based on finding a convex decomposition for an arbitrary state $\rho$ in terms of states with vanishing entanglement and at most one state with non-vanishing entanglement. However, for one-root states, such a decomposition is possible only in one way: that is, into a pair formed by the root $\ket{z}$ and some other state $\ket{\omega}$;  hence the entanglement of $\rho$ is trivially given by  $E(\ket{\omega})$ with the corresponding weight. Additionally, bounds which use methods such as the conjugate gradient \cite{audenaert_2001, rothlisberger_2009} are also guaranteed to converge to the right value. By the Schr\"{o}dinger--Hughston--Jozsa--Wootters theorem \cite{schrodinger_1936, hughston_1993}, any two decompositions for a given density matrix $\rho$ are related by applying  a unitary matrix, therefore a typical instance of a numerical method of this kind calculates the gradient for a given entanglement measure on the unitary manifold and uses it to reach the minimum in the convex roof. For one-root states, however, any choice of the initial decomposition gives the right entanglement value by Eq.~\eqref{affine}, and the value of the gradient of $E$ on the unitary manifold can be verified (numerically) to stay uniformly zero.

In conclusion, we have shown that every polynomial entanglement measure $E$ of degree 2 is affine for any rank-$2$ state $\rho$ for which there is only one pure state $\ket{z}$ in the range of $\rho$ such that $E(\ket{z})=0$. This renders calculating the convex roof of $E$ trivially easy in any such case, as the entanglement of $\rho$  does not depend on its pure-state decomposition.
The method applies to many significant mixed states which did not enjoy known formulae before, as is the case for the three-tangle of the marginals of several classes of four-qubit pure states.

The results of Theorem~\ref{geothm} can be used for  evaluation of various polynomial generalizations of the tangle in four and more qubits \cite{osterloh_2004,gour_2010}  whose states obey the one-root property; explicit instances can be readily constructed in analogy to the ones reported here for two and three qubits. Moreover, the geometric approach presented herein is rather powerful and applicable also to higher-dimensional systems using a generalized Bloch vector approach \cite{kimura_2003}, although the properties of the complex polynomials encountered in the definition of the entanglement measures do not seem to allow for a simple generalization of the concept of one-root states. A possible extension of this work would be to find classes of qudit states with equivalent properties, which might lead us to accomplish an even more comprehensive study of multipartite entanglement.

\begin{acknowledgments}We thank the European Research Council (ERC) Starting Grant GQCOP (Grant No.~637352), for financial support. We acknowledge fruitful discussions with J. Louko, A. Streltsov, A. Winter, W. K. Wootters, and especially K. Macieszczak.\end{acknowledgments}


%

\end{document}